\documentclass[smallextended]{svjour3}

\usepackage[margin=1in]{geometry}
\usepackage{latexsym,graphicx,amssymb,amsmath}

\usepackage[dvipsnames,usenames]{color}

\usepackage{tikz}
\usepackage{tkz-euclide}
\usepackage{empheq}
\usepackage{subcaption}
\usetikzlibrary{arrows.meta,calc,decorations.markings,math,arrows.meta}
\usetikzlibrary{angles}
\usetikzlibrary{backgrounds}
\usetikzlibrary{shapes.misc}
\usetikzlibrary{patterns,quotes}

\newcommand{\LSA}{{\operatorname{LSA}}}

\smartqed

\title{The Tree of Blobs of a Species Network: Identifiability under the Coalescent} 
\author{Elizabeth S. Allman \and Hector Ba\~nos \and Jonathan D. Mitchell \and John A. Rhodes }

\institute{ Elizabeth S. Allman \and John A. Rhodes \at Department of Mathematics and Statistics, University of Alaska Fairbanks, Fairbanks, AK, 99775, USA
	\and
	Hector Ba\~nos \at Department of Biochemistry \& Molecular Biology, Faculty of Medicine, Dalhousie University, Halifax, Nova Scotia, CANADA, and\\
	Department of Mathematics and Statistics, Faculty of Science, Dalhousie University, Halifax, Nova, Scotia, CANADA
	\and
	Jonathan D. Mitchell \at  Department of Mathematics and Statistics, University of Alaska Fairbanks, Fairbanks, AK, 99775, USA, and \\
	School of Natural Sciences (Mathematics), University of Tasmania, Hobart, TAS 7001, AUSTRALIA, and\\
ARC Centre of Excellence for Plant Success in Nature and Agriculture, University of Tasmania, Hobart, TAS 7001,\\ AUSTRALIA
	}
	
\begin{document}
 	
 \maketitle
 
\begin{abstract} Inference of  species networks from genomic data under the Network Multispecies Coalescent Model
is currently severely limited by heavy computational demands. It also remains unclear how complicated networks can be for consistent inference to be possible. As a step toward inferring a general species network, this work considers its tree
of blobs, in which non-cut edges are contracted to nodes, so only
tree-like relationships between the taxa are shown. An identifiability theorem, that most 
features of the unrooted tree of blobs can be determined from the distribution of gene quartet topologies, is established.
This depends upon an analysis of gene quartet concordance factors under the model, together with a new combinatorial
inference rule. The arguments for this theoretical result suggest a practical algorithm for tree of blobs inference, to be fully 
developed in a subsequent work.
\end{abstract}

\section{Introduction}

Methods for inference of evolutionary relationships between organisms
are well-developed provided those relationships can be adequately
described by a tree. If hybridization or some form of lateral
gene transfer has occurred, tools for data analysis are much more
limited. An essential complication is that when
such gene transfer has occurred between closely related taxa, the population-genetic effect of
incomplete lineage sorting is also likely. Thus individual gene
relationships may conflict with the primary tree-like
species relationships (if some can be considered to be primary) due to the intermixed effect of
these two processes.

The appropriate stochastic model to capture these processes is the
Network Multispecies Coalescent (NMSC).  Under the NMSC combined
with standard sequence substitution models, Bayesian methods for
inference of species networks have been implemented (BEAST 2/SpeciesNetwork \cite{Zhang2017}, 
PhyloNet \cite{ZhuEtAl2016,ZhuEtAl2018}, BPP \cite{Yang2019}). However,
they are limited by computational demands to small data sets
of few taxa and few genes. Pseudolikelihood methods that treat
inferred gene trees as data are able to handle larger data sets
(PhyloNet \cite{YuNakhleh2015}, SNaQ \cite{Solis-Lemus2016}), but 
require prespecification of the number of
reticulation events, with at best heuristic assessment of that
number. In addition, to reduce computational effort, inference may be limited to the class of level-1
networks, though a biological
justification for that may be lacking. 
A final approach starting with
inferred gene trees combines statistical tests for small networks with
combinatorial methods to assemble a large network (NANUQ \cite{ABR2019}). This is
considerably faster and offers some insight into model fit, but also is
currently limited to level-1 structure.

It is not known how complex a species network can be
for its inference from specific data types to be even theoretically possible. This is
the question of identifiability of the network (either topological or
metric) under the NMSC model: Does the distribution of observations under the NMSC uniquely determine the network? The most complete result in the level-1  topological case comes from
Ba\~nos' study  of identifiability from quartet concordance factors \cite{Banos2019}. Using different notions of data, however,
several works have 
studied the identifiability question for general
networks without the coalescent. Researchers have, for instance,  investigated what can be determined from average intertaxon distances on a network \cite{XuAne2021}, as well as shortest distances and distance multisets  \cite{VanIersalEtAl2020}. Identifiability from induced 4-taxon networks \cite{HuberEtAl2018}, rooted 3-taxon networks \cite{SempleToft2021}, and counts of paths from interior nodes to taxa \cite{ErdosEtAl2019} have also been explored, among other notions.
 
\medskip

In this work we approach the network inference problem from a different direction, trying to determine 
only the tree-like evolutionary relationships for a collection of taxa, hence isolating the parts of their history when more
complicated network features are formed. More formally, we study
the \emph{tree of blobs} of the network \cite{Gusfield2007}, a tree
in which each group of edges in the network describing complex gene
transfer, i.e., each blob, has been shrunk to a single node. (A closely related notion appears in \cite{MurakamiEtAl2019}.) The tree
of blobs thus shows all tree-like parts of the network, and its
inference could be useful to researchers who may subsequently focus on
inferring the structure of each blob by other methods.

Our goal here is to show the topology of the unrooted tree of blobs for a network is identifiable from gene quartet data
under the NMSC model. That is, the distribution of gene quartet topologies arising 
under the NMSC on a fixed species network uniquely determines the unrooted
tree of blobs of that network. We make no assumptions on blob structure, but do require that numerical parameters lie outside an exceptional set of measure zero.
Thus consistent inference of the tree of blobs is theoretically possible.

We first study  the probabilities of
quartets displayed across independent gene trees under the NMSC, under a generic assumption on numerical parameters. These probabilities --- the quartet concordance factors ($CF$s) --- allow for the identification of some sets of 4
taxa that must be collectively related through a blob, while proposing a
resolved quartet tree topology for others. 
A new combinatorial inference rule is then developed
that allows this information to be used to identify additional sets of four taxa related through a single blob, even though their $CF$s suggested otherwise. We show that repeated application of this rule yields
all sets of four taxa with blob relationships. Then, with all such blob quartets known, and tree topologies
assigned to other sets of four taxa, by treating blob quartets as unresolved we obtain complete information on all quartets displayed on the tree of blobs. 
This information is enough to determine the tree of blobs  \cite{Semple2005,Rhodes2020}.

Although rules for inference of large networks from 4-taxon networks have been considered previously \cite{HuberEtAl2018}, our rule is different in purpose. It neither assumes knowledge of the full 4-taxon blob structure, nor attempts to infer detailed blob structure on a larger network. Earlier work on quartet closure rules for trees, surveyed in \cite{GrunewaldHuber2007}, is also similar in spirit to the rule developed here. 

Our approach suggests an algorithm for tree of blobs inference that will be fully developed in
a subsequent paper focused on data analysis. First a statistical test  can be applied to gene quartet counts to detect blob and tree relationships on induced 4-taxon networks.
Then the inference rule is applied repeatedly, until no new blob relationships on the full network are inferred. Finally, the quartet intertaxon distance \cite{Rhodes2020} is computed treating blob relationships as unresolved. A standard distance-based tree building algorithm, such as Neighbor-Joining \cite{NJ87}, then yields an estimate of the tree of blobs.
This is broadly similar to the steps in NANUQ \cite{ABR2019} for inference of a level-1 network, but
the inference rule step is new, and the distance, which in principle should fit a tree, does not require an analysis by NeighborNet \cite{Bryant2004} or construction of a splits graph \cite{Dress2004}.

Many methods have been developed for a more detailed detection of hybridization or gene transfer than the tree of blobs depicts, e.g. \cite{BlischakEtAl2018,GreenEtAl2010,HamlinEtAl2020,HibbinsHahn2022}. Once the tree of blobs has been inferred for a collection of taxa, such methods might be applied to a subset of the taxa in order to explore the structure of a blob through a finer analysis. Unfortunately, these methods are generally restricted to a small number of taxa, and simple scenarios (e.g., level-1). Much work remains to be done to both expand the scope of methodology for inferring blob structure, and to delineate both theoretical and practical limits to its inference. 

\medskip

Our presentation is structured as follows. Section \ref{sec::networksmodels}
provides basic definitions and background on the NMSC model. In
Section \ref{sec:Blobquartets4} we prove the fundamental result that
from quartet concordance factors under the NMSC on a 4-taxon network
one can determine whether the taxa are related through a single
blob (i.e., a 4-blob), or not. If not, then all displayed trees on the 4-network have the same tree topology, which can also be determined.
Establishing these facts requires an analysis based in the NMSC
model.  In Section \ref{sec:BlobquartetsP}, we use combinatorial
arguments to show that from such information on the 4-taxon induced
subnetworks of a larger network we can, through certain inference
rules, gain information on all larger blobs.  
Section \ref{sec:main} quickly completes the argument for identifiability, and sketches the algorithm 
for tree of blobs inference suggested by the proof.

 \section{Networks and models}\label{sec::networksmodels}

\subsection{Phylogenetic networks}\label{sec::networks}

The Network Multispecies Coalescent model of gene tree formation
within a species network underlies this work, so we give an
appropriate definition of a phylogenetic network for that model.

\begin{definition} \label{def::network}
  \cite{Solis-Lemus2016,Banos2019} A \emph{topological rooted binary
    phylogenetic network} $\mathcal{N}^+$ on taxon set $X$ is a
  connected directed acyclic graph with nodes $V$ and edges $E$,
  where $V$ is the disjoint union $V = \{r\} \sqcup V_L \sqcup V_H
  \sqcup V_T$ and $E$ is the disjoint union $E = E_H \sqcup E_T$, together with
  a bijective leaf-labeling function $f : V_L \to X$ with the
  following characteristics:
	\begin{itemize}
		\item[1.] The \emph{root} $r$ has in-degree 0 and out-degree 2.
		\item[2.] A \emph{leaf} $v \in V_L$ has in-degree 1 and out-degree 0.
		\item[3.] A  \emph{tree node} $v\in  V_T$ has in-degree 1 and out-degree 2.
		\item[4.] A \emph{hybrid node} $v\in  V_H$ has in-degree 2 and out-degree 1.
		\item[5.] A \emph{hybrid edge} $e=(v,w) \in E_H$ is an edge whose child node $w$ is hybrid.
		\item[6.] A \emph{tree edge} $e=(v,w) \in E_T$ is an edge whose
		child node $w$ is either a tree node or a leaf.
	\end{itemize}
\end{definition}

See Figure \ref{fig::network}(L) for an example of a rooted binary phylogenetic network. In that figure, and in others throughout this work, red indicates hybrid nodes and the hybrid edges leading to them. 

\begin{figure}
\begin{center}
	\includegraphics{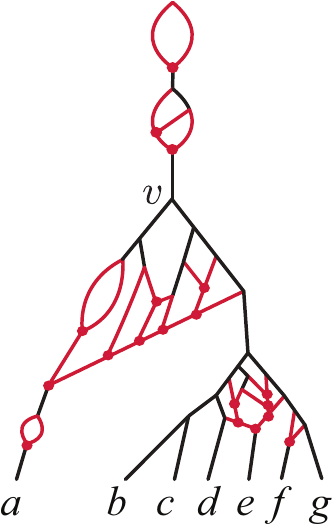} \hskip .9cm
	\includegraphics{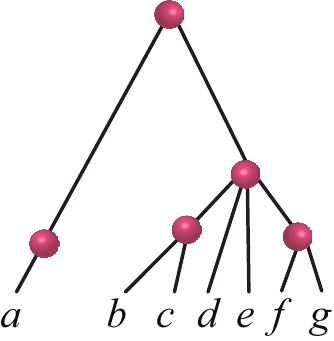} \hskip .9cm
	\includegraphics{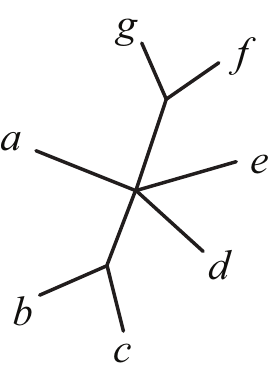}
\end{center}
\caption{(L) A species network $\mathcal N^+$, with edge lengths in
  coalescent units. Red indicates hybrid nodes and hybrid edges. The lowest stable ancestor (LSA) of the network is $v$. This network has 6 non-trivial blobs (a 5-blob, two 3-blobs, and three 2-blobs), and a single trivial 3-blob. (C) The tree-like structure of the LSA network $\mathcal N^\oplus$, obtained by deleting parts of the network above the LSA $v$, and showing blobs as red spheres. A sphere is used to suggest an unknown and potentially complicated blob structure. (R) The reduced unrooted tree of blobs, $T_{rd}(\mathcal
  N^-)$, obtained by shrinking blobs in the LSA network to nodes, unrooting, and suppressing degree-2 nodes. }\label{fig::network}
\end{figure}

\begin{definition}
A \emph{cut edge} in a graph is one whose deletion increases the number of connected components.
\end{definition}

Note that the notions of cut and non-cut edges are not the same as tree and hybrid edges.
Although a hybrid edge is never a cut edge, tree edges may or may not be cut edges. For instance, in Figure \ref{fig::network}(L), the child edges of $v$ are both tree edges and non-cut, while the parent edge of $v$ is tree and cut.

Edge directions on a rooted phylogenetic network induce a partial order on its
nodes. We say that a node $u$ is \emph{above} or \emph{ancestral to} a node
$v$, or $v$ is \emph{below} or \emph{descended from} $u$, if there is
a directed path in the network from $u$ to $v$. Thus the root is above
all other nodes. We use the same terms to refer to similar
relationships between edges, or between edges and nodes.

\smallskip

A topological network is one parameter of the
NMSC model. Additional numerical parameters are introduced by giving the
network a metric structure.  Edge lengths are measured in
\emph{coalescent units} (units of generations/population size). In
addition, we specify probabilities that a gene lineage at a hybrid
node follows one or another hybrid edge as it traces back in time
toward the network root. 

\begin{definition} A \emph{metric rooted binary phylogenetic network}
  $( \mathcal N^+, \{\ell_e\}_{e\in E},\{\gamma_e\}_{e\in E_H})$
  is a
  topological rooted binary phylogenetic network together with an
  assignment of weights or \emph{lengths} $\ell_e$ to all edges and
  \emph{hybridization parameters} $\gamma_e$ to all hybrid edges
  subject to the following restrictions:
	\begin{itemize}
		\item[1.] The length $\ell_e$ of a tree edge $e\in E_T$ is positive.
		\item[2.] The length $\ell_e$ of a hybrid edge $e\in E_H$ is non-negative.
		\item[3.] The hybridization parameters $\gamma_e$  and $\gamma_{e'}$ for a pair of hybrid edges $e,e'\in E_H$ with the same child hybrid node are positive and sum to 1.
	\end{itemize}
\end{definition}

Our use of the term hybridization parameter does not imply the NMSC model only applies to describing hybridization in any strict biological sense; it is simply a convenient shorthand for a parameter quantifying gene flow. In some works these parameters are called \emph{inheritance probabilities} \cite{Solis-Lemus2016}.

Note that we require tree edges to have positive length, since lengths of zero would effectively allow networks to be non-binary. Since zero lengths are non-generic in the parameter space, our formal statements of results holding for generic parameters would need no modification if they were allowed, though perhaps they would be more open to misinterpretation. We
do explicitly allow hybrid edges to have length 0,
to model possibly instantaneous jumping of a lineage from one
population to another. A careful reading of our arguments shows that while such values are also non-generic, they do not lead to additional points in the exceptional set of non-generic points where our claims fail.

\smallskip

The following analog of the most recent common ancestor of taxa on a tree is needed.

\begin{definition}\label{def::LSA}\cite{Steel2016}
  Let $\mathcal{N}^+$ be a (metric or topological) rooted binary
  phylogenetic network on $X$ and let $Z\subset V$ be any nonempty
  subset of the nodes of $\mathcal{N}^+$. Let $D$ be the set of nodes
  which lie on every directed path from the root $r$ of
  $\mathcal{N}^+$ to any $z\in Z$. Then the \emph{lowest stable
    ancestor (LSA) of $Z$ on $\mathcal{N}^+$}, denoted
  \emph{LSA}$(Z,{\mathcal{N}^+})$, is the unique node $v\in D$ such
  that $v$ is below all $u\in D$ with $u\neq v$. 
  
  The LSA of the network, $LSA(\mathcal N^+)$, is the LSA of its leaves, $LSA(V_L,\mathcal N^+)$.
\end{definition}

As shown in Figure \ref{fig::network}(L), a rooted phylogenetic network
may have a complex structure above its LSA. (If the
network is level-1, this is a chain of 2-cycles, as discussed in
\cite{Banos2019}.)  Since our methods based on gene quartets do not
give us any information about structure above the LSA, we focus only
on the structure below the LSA, sometimes with edge direction information lost.

To formalize this,
by suppressing a
node with both in- and out-degree 1 in a directed graph we mean
replacing it and its two incident edges with a single edge from its
parent to its child. Suppressing a degree-2 node between two
undirected edges means replacing it and its two incident edges with a
single undirected edge. Suppressing a node between an undirected edge and a directed
out-edge means replacing it and its two incident edges with a
single edge with the out-edge direction. Suppressing a node between a directed in-edge
and an undirected edge
means replacing it and its two incident edges with a
single undirected edge.
In all these situations, for a metric graph the new edge is assigned a
length equal to the sum of lengths of the two replaced. If the out-edge
was hybrid, the new edge is also hybrid and retains the hybridization
parameter.

\begin{definition}\label{def::undirected} \cite{Banos2019}
Let $\mathcal{N}^+$ be a (metric or topological) rooted binary
phylogenetic network on $X$.
\begin{enumerate}
\item  The \emph{LSA network} $\mathcal N^{\oplus }$ induced from $\mathcal N^+$ is the network
obtained by deleting all edges and nodes above $\LSA( \mathcal{N}^+)$, and designating $\LSA(\mathcal{N}^+)$ as
the root node.		
\item  The \emph{semidirected unrooted network} $\mathcal N^-$ is the unrooted network
obtained from the LSA network $\mathcal{N}^\oplus$ by undirecting all tree edges and suppressing the root,
but retaining directions of hybrid edges.		
\end{enumerate}
\end{definition}

We often need to pass to a network on a subset of taxa from one on a larger set.

\begin{definition} \label{def:induced} Let $\mathcal{N}^+$ be a
(metric or topological) rooted binary phylogenetic network on $X$
and let $Y\subset X$. The \emph{induced rooted binary network}
$\mathcal{N}^+_{Y}$ on $Y$ is the network obtained from
$\mathcal{N}^+$ by retaining only those nodes and edges ancestral to
one or more taxa in $Y$, and then suppressing all nodes with both
in- and out-degree 1.  We then say $\mathcal{N}^+$ \emph{displays}
$\mathcal{N}^+_{Y}$.
\end{definition} 

\subsection{Cycles, blobs, and  quartets} \label{subsec:blobquartets}

Since rooted phylogenetic networks are acyclic by definition, we use
the word \emph{cycle} to refer to a sequence of edges in the network
which forms a cycle when all edges are undirected.  
A \emph{$k$-cycle} is a cycle composed of $k$ edges.
 
Although we focus on phylogenetic networks, the following definition applies more broadly.

\begin{definition} A \emph{blob} on a network is a maximal connected
 subnetwork that has no cut edges.  A blob is \emph{trivial} if it consists of a single node. An edge in the 
 network is said to  be \emph{incident} to a blob if exactly one of its incident nodes is in
the blob.  A blob has \emph{degree $m$} or is an \emph{$m$-blob} if a) it has has exactly $m$ cut edges
incident to it and the network's root is not in the blob, or b) it has exactly $m-1$ cut edges incident to it and the root is in the blob. 
\end{definition}

We define an $m$-blob in this way for two reasons: First, it results in
the degree of the blob containing the LSA not changing in passing from a rooted network $\mathcal N^+$ to its LSA network $\mathcal N^\oplus$. Second, the NMSC model considers an ``above the root" population of infinite duration in which lineages may coalesce. This is essentially an additional edge, of infinite length, incident to the root. In our terminology if the root of a binary network is a trivial blob, then it is a degree-2 node but forms a degree-3 blob. 
  
A network's blobs can equivalently be defined as the 2-edge-connected components \cite{XuAne2021}, or as the connected components obtained by deleting all cut edges in the network.

On a rooted binary phylogenetic tree, leaves are the only 1-blobs, while the root 
and internal nodes are 3-blobs. On a non-binary tree,
polytomous nodes are $k$-blobs with $k\ge 4$.  Non-tree
phylogenetic networks may have $k$-blobs that are not nodes for any
$k> 1$. The simplest blobs have the form of cycles, and a network with only such blobs is level-1. In general,
however, blob structure may be much more complicated, with a few simple examples shown in Figures \ref{fig::network}(L) and \ref{fig::blobs}.
 
 \begin{figure}[h]
	\begin{center}
		\includegraphics{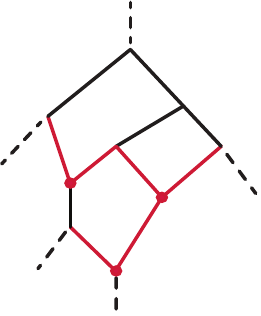} \hskip 1cm
		\includegraphics{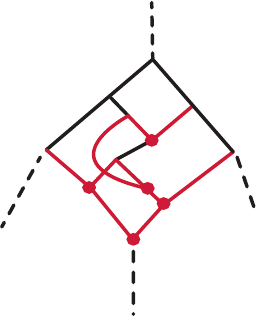} \hskip 1.5cm
		\includegraphics{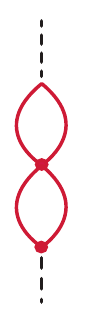}
	\end{center}
       \caption{Examples of blobs in networks. Red indicates hybrid nodes, and hybrid edges above them. Cut edges 
       	incident to the blobs are represented by dotted line segments:
      (L) a planar 5-blob, (C)) a non-planar 4-blob, (R) a single
    2-blob in a non-binary network, formed from two 2-cycles
     sharing a single node.}\label{fig::blobs}
\end{figure}

\medskip

As is well known, on a tree any 3 taxa determine a unique node where undirected paths between each pair of taxa meet, or equivalently a node whose deletion leaves the taxa in distinct connected components. If the tree is not binary, larger sets of taxa may or may not  determine a node in this way. The following definition formalizes a similar notion for networks.
 
\begin{definition}
A blob is \emph{determined by} a set of leaf labels $S$ with $|S|\ge3$ if deletion
  of the cut edges incident to the blob leaves the elements of $S$ in distinct connected components.
\end{definition}

On a network $\mathcal N^+$ every subset of 3 taxa 
determines a blob, and every $m$-blob with $m\ge 3$ that is below the LSA of $\mathcal N^+$ is determined by one or
more subsets of 3 taxa. Blobs above the LSA are not determined  by any subset of taxa, while an $m$-blob containing the LSA is
determined by 3 taxa if $m\ge4$.

A set of $k\ge 4$ taxa may or may not determine a
blob, but if it does it must be an $m$-blob with $m\ge k$.  For instance, the network of Figure \ref{fig::BTquartets}
has a 5-blob determined by the sets $\{a,b,c\}$,  $\{a,b,d,f\}$, and others. The set $\{a,b,e,f\}$, however, does not determine a blob.

\begin{figure}
\centering
\includegraphics[width=2in]{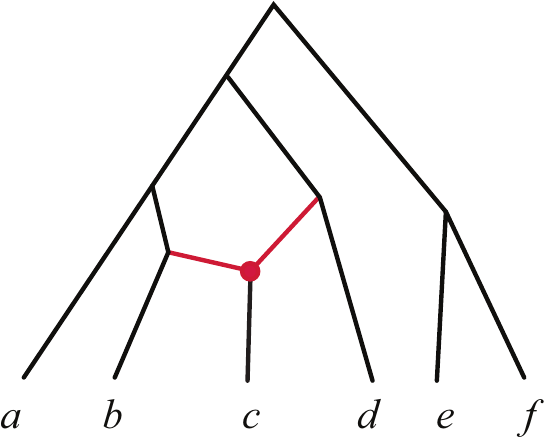}
\caption{A network with a 5-blob determined by the sets $\{a,b,c\}$,  $\{a,b,d,f\}$, 
	and other sets. The set $\{a,b,e,f\}$, however, does not determine a blob. 
	Both $\{a,b,c,d\}$ and $\{a,b,d,e\}$ are B-quartets on this network. 
	While $\{a,b,c,d\}$ is also a B-quartet on its induced 4-taxon network, 
	$\{a,b,d,e\}$ is a T-quartet on its induced 4-taxon network.}\label{fig::BTquartets}
\end{figure}

Note that our definition of blob differs slightly from that given by
\cite{Gusfield2007}, in which a blob is a maximal set of edges formed
by recursively including cycles sharing at least one edge with an
earlier cycle. By that definition, if two cycles share only a node as in Figure \ref{fig::blobs}(R),
they would be considered to be 2 distinct blobs. In contrast, they form a
single blob under our definition. In \cite{Gusfield2007}, this situation is handled by
inserting an edge to separate two such cycles, joining each at the
node they formerly shared, thus making edge-disjoint cycles also
node-disjoint. Our restriction to binary networks
rules out this possibility regardless.

\begin{definition} A \emph{chain of blobs} in a  network is a subnetwork composed of a sequence of 1- and 2-blobs connected by their incident edges.
\end{definition}

This notion generalizes the chain of 2-cycles defined for level-1 networks in
\cite{Banos2019}.  A chain of blobs will have  1- or 2-blobs at its ends, but all other blobs in the chain will be 2-blobs.
Just as a level-1 phylogenetic network may have a chain of 2-cycles above its LSA, a general phylogenetic 
network will have a (possibly empty) chain of blobs as the subnetwork between its root
and LSA, as in Figure \ref{fig::network}(L).

\begin{definition} \cite{Gusfield2007} The \emph{tree of blobs},
  $T(\mathcal N)$, for a general connected network, $\mathcal N$, is
  the tree obtained by contracting each blob to a node, that is, by
  removing all of the blob's edges and identifying all its nodes.
If the network is rooted, the tree of blobs remains rooted at the
  same node, or the one arising from identifying the original root
  with other nodes.
\end{definition}

 An equivalent construction of a blob tree in \cite{XuAne2021} has nodes for each blob in $\mathcal N$, with edges connecting them if there is an edge  with endpoints in the two blobs in $\mathcal N$.  
 
The tree of blobs is generally not binary, even when the network is. A blob with $m$ incident cut edges  in a network produces an $m$-multifurcation in its tree of blobs.
 Nodes of degree 4 or more in the tree of blobs indicate non-trivial blobs for a binary network, while those of degree 2 or 3 may correspond to trivial or non-trivial blobs in the network.

\smallskip

While this definition of a tree of blobs applies to an arbitrary
connected network, a slight variant is more useful here, as only
some of the features of the tree of blobs for a species network may be
identified by our methods and data. Any 2-blobs become
nodes of degree 2 in the tree of blobs, but we will suppress these
since we cannot detect them. Also, while we cannot detect any structure between the root of the network and its LSA, even after
suppressing nodes of degree 2 arising from blobs above the LSA, an undetectable edge above the LSA might remain. We therefore discard this as well.

\begin{definition}
  The \emph{reduced rooted tree of blobs}, $T_{rd}(\mathcal N^+)$, of
  a rooted phylogenetic network $\mathcal N^+$ is obtained from the
  tree of blobs $T(\mathcal N^\oplus)$ of the LSA network by
  suppressing all nodes of both in-degree 1 and out-degree 1. The
  \emph{reduced unrooted tree of blobs}, $T_{rd}(\mathcal N^-)$, of
  $\mathcal N^+$ is obtained from the tree of blobs $T(\mathcal N^-)$
  of the unrooted semidirected network by suppressing all nodes of
  degree 2.
\end{definition}

See Figure \ref{fig::network} for an example of a network and its reduced unrooted tree
of blobs.  The reduced unrooted tree of blobs  $T_{rd}(\mathcal N^-)$ is undirected since the only directed edges in
$\mathcal N^-$ are hybrid edges, which are in blobs, and thus lost when passing to its tree of blobs.

The reduced unrooted tree of blobs $T_{rd}(\mathcal N^-)$ can also be
obtained from the rooted one $T_{rd}(\mathcal N^+)$ by undirecting all edges, and either
suppressing the root if it has degree 2 (as a node) or dropping its designation as
the root if it has larger degree.

Note that if the  $\LSA$ of the original phylogenetic network $\mathcal N^+$
lies in an $m$-blob, $m\ge 4$, that blob gives only an $(m-1)$-multifurcation in the
reduced unrooted tree of blobs. If the $\LSA$ lies in a 3-blob, then that blob will
be completely suppressed, and not represented by a node.

\medskip

We next introduce terminology to express the relationships a set of four taxa
might have to the blob structure of a network. We follow the standard
convention of using the word \emph{quartet}  to
mean a particular unrooted binary topological tree on four taxa. For
instance the quartet $ab|cd$ is the topology with cherries $\{a,b\}$
and $\{c,d\}$ separated by an internal edge.  The \emph{unresolved quartet}
is the star topology for the 4-taxon tree, denoted $abcd$. The
following additional terminology is also useful in the network setting.

\begin{definition}\label{def:Bquartet} A set of four taxa $Q=\{a,b,c,d\}$
  on an $n$-taxon phylogenetic network is a \emph{Blob quartet}, or
  \emph{B-quartet}, if there is a blob on the network which is determined
  by $Q$. \end{definition}

Equivalent conditions for  $Q=\{a,b,c,d\}$ being a B-quartet are 1) the deletion of all edges in
  a single blob leaves the elements of Q in four distinct connected
  components, and 2) the unresolved quartet $abcd$ is displayed on the tree of blobs
  $T_{rd}(\mathcal N^-)$.
   The blob referred to here may be an $m$-blob for any
  $m\ge 4$.

If $\{a,b,c,d\}$ is not a B-quartet, then in the tree of blobs there must be an edge whose deletion disconnects two of these
taxa from the others. Consequently, the tree of blobs displays a resolved
quartet tree for these taxa.

\begin{definition}\label{def:Tquartet}
  If a set of four taxa is not a B-quartet on an $n$-taxon phylogenetic network, $n\ge 4$, then it
  is a \emph{tree-like quartet}, or \emph{T-quartet}. The resolved
  quartet \emph{associated} to a T-quartet is that displayed on the tree of
  blobs $T_{rd}(\mathcal N^-)$.
\end{definition}

Note that the induced 4-taxon network on a T-quartet need not be a
tree, since the induced network on the four taxa may contain non-trivial
2-blobs and 3-blobs. However there can be no larger blobs. Nonetheless
this induced network is ``tree-like'' in the sense that it will have
a cut edge whose removal disconnects the four taxa into two groups of 2. 
Equivalently, every tree displayed on the 4-taxon network has the same 
resolved quartet topology. Thus any T-quartet on a large network is
also a T-quartet on the induced quartet network.

In contrast, the induced network on a B-quartet may or may not
have a 4-blob, and can even be a tree. In passing from a network to an induced network on fewer taxa,
blobs may split into smaller blobs, and in some cases reduce to
tree-like relationships. Indeed, this happens even in the level-1 case
with a single cycle of $k$ edges, $k\ge 5$. In Figure
\ref{fig::BTquartets}, for instance,  $\{a,b,d,e\}$ is a B-quartet in the full network, yet becomes a T-quartet
on the
induced 4-taxon network. 
However, $\{a,b,c,f\}$ is a B-quartet on both the full and the induced
networks. 
	
\subsection{Coalescent model on networks and quartet concordance factors}\label{ssec:NMSC} 
 
The formation of gene trees, tracking the ancestral relationships of
individual lineages within populations of ancestral species, is
governed not only by the relationships of those species, but also
population-genetic effects. Going backwards in time, these lead to
gene lineages merging not when they first enter a common ancestral
species, but rather further in the past. If they fail to merge before
entering an ancestral population with yet other lineages, the gene
tree relationships that form may differ from the species
relationships.  When the species relationships are described by a tree
rooted at a common ancestor, the \textit{multispecies coalescent
(MSC)  model}  is the standard probabilistic description of gene tree
formation capturing this process \cite{Pamilo1988,MSCpaper2009}.

The \textit{network multispecies coalescent  (NMSC) model} 
\cite{Meng2009,Nakhleh2012,ZhuEtAl2016} generalizes the MSC,
allowing a finite number of hybridization events, or other discrete
lateral gene transfer events, between ancestral populations.  Its
parameters are captured by a metric, rooted phylogenetic network,
assumed here to be binary, as in Definition \ref{def::network}. Edge
lengths are given in coalescent units (computed as number of
generations/population size), so that the rate of coalescence between
two lineages is 1.  At a hybrid node in the network, a gene lineage
may pass into either of two ancestral populations, with probabilities
given by the hybridization parameters $\gamma, 1-\gamma$ for the hybrid 
edges.  This differs from other generalizations of the MSC, such as the
structured coalescent, where gene flow may be continuous over a time
interval.

The NMSC model determines a distribution of binary metric gene trees,
and, through marginalization, distributions of binary topological gene
trees on subsets of taxa. In this work we use only one type of
marginalization, to unrooted binary topological gene trees on subsets
of four taxa, or \emph{gene quartets}.  The probability of a gene quartet
is thus a function of the metric species network parameters under the
NMSC.  Formulas for these probabilities were obtained in the tree case
in \cite{Allman2011}, and for level-1 networks in
\cite{Solis-Lemus2016}, with further study in \cite{Banos2019}. Here
we do not restrict to level-1 networks, and without any
assumptions on blob structure one cannot obtain precise
formulas for gene quartet probabilities. Nonetheless, some features of
these probabilities can be analyzed sufficiently for application to determining
the tree of blobs of the network.

\begin{definition}
  Let $\mathcal N^+$ be a metric rooted binary phylogenetic network on a taxon set $X$,
  and $a,b,c,d\in X$ distinct taxa.  Then for the gene quartet
  $ab|cd$, the \emph{quartet concordance factor} $CF_{ab|cd} =CF_{ab|cd}
  (\mathcal N^+)$ is the probability under the NMSC on $\mathcal N^+$
  that a gene tree displays the quartet $ab|cd$. The \emph{quartet
    concordance factor for  taxa $a,b,c,d$}, or more simply the \emph{concordance factor}, is the ordered triple
  $$CF_{abcd} =CF_{abcd}(\mathcal N^+)=(CF_{ab|cd},CF_{ac|bd},CF_{ad|bc})$$
	of concordance factors of each quartet on the taxa.
\end{definition}

Since under the NMSC gene trees are binary, and all gene tree topologies have positive probability, the entries of $CF_{abcd}$ are positive and sum to 1. Note that permuting $a,b,c,d$ permutes the entries of
$CF_{abcd}$. Nonetheless, when $a,b,c,d$ are clear from context, such
as when $|X|=4$, we write $CF$ for $CF_{abcd}$.

In \cite{Allman2011} it was shown that if the species network is a
tree then two of the three entries of $CF_{abcd}$ must be equal, with
the third no smaller. We need the following broader notion.

\begin{definition} \label{def:cutCF} The concordance factor $CF_{abcd}$ is a \emph{cut $CF$} if two
  of its entries are equal, and \emph{strictly cut} if in addition the
  third is distinct.  If $CF_{abcd}$ is strictly cut with
  $CF_{ab|cd}\neq CF_{ac|bd}=CF_{ad|bc}$, then we say $CF_{abcd}$ is
  \emph{strictly $(ab|cd)$-cut}. If $CF_{abcd}$ is not cut, we say it
  is \emph{non-cut}.
\end{definition}

The term ``cut" is motivated by Theorem \ref{thm:CFdetect}
of the next section, which states that for generic parameters a $CF$ is cut exactly
when there is a cut edge in the 4-taxon network whose deletion from
the network leaves two connected components each with two taxa.

\smallskip

We emphasize that Definitions \ref{def:Bquartet} and
\ref{def:Tquartet} of B- and T-quartets refer to the relationship of 4
taxa through the topology of a specified network, while Definition \ref{def:cutCF}
of cut and non-cut $CF$s refers to properties of the probability
distribution under the NMSC. In passing to an induced network, B-quartets may become T-quartets, although $CF$s remain unchanged.
 
 Theorem \ref{thm:CFdetect} below shows that on 4-taxon networks there
is a close correspondence between B-quartets and non-cut $CF$s. However, these notions are more subtly related on larger networks. For the
network of Figure \ref{fig::BTquartets}, for instance, $\{a,b,d,e\}$ is a B-quartet yet has a
strictly cut $CF$. This issue is the main obstacle to showing identifiability of
the tree of blobs, to be overcome with Theorem \ref{thm:infrule} below.

\section{Blob quartet identifiability on 4-networks}\label{sec:Blobquartets4}

We work under the NMSC model, so that specification of model
parameters through a metric rooted binary phylogenetic network
determines a distribution of $n$-taxon gene trees, and by
marginalization, the theoretical quartet $CF$s for each subset of four
taxa.

Although our ultimate goal is to identify the reduced unrooted tree of
blobs of a rooted phylogenetic network from the $CF$s, with no
assumption on level or other particular network structure, our approach to
doing this is by first determining B-quartets. In this section we
show that by applying certain inference rules,  all
B-quartets on 4-taxon networks can be identified from the $CF$s, assuming generic values of numerical
parameters. 

By \emph{generic} numerical parameters we mean all those that lie outside of a subset of measure zero in the
parameter space. While we
do not give an explicit description of such an exceptional set, a good
intuitive description that it has measure zero is that
if parameter values were chosen at random from an absolutely continuous distribution,
then with probability 1 they would not be exceptional. For complex
stochastic models it is quite common for identifiability results to
depend upon the exclusion of some ``small'' exceptional subsets
of the parameter space \cite{AMR2009}.

\medskip

A basic combinatorial observation, whose proof we omit, is the following.
\begin{lemma} \label{lem:4netblobs} Let $\mathcal N^+$ be a 4-taxon rooted binary phylogenetic network. Then the semidirected  unrooted network $\mathcal
  N^-$ must have either
\begin{enumerate} 
\item exactly one 4-blob, or
\item  exactly two 3-blobs.
\end{enumerate}
In either case, $\mathcal
  N^-$ may have any number of 2-blobs, but no other non-leaf
blobs. In case 1, the reduced unrooted tree of blobs
$T_{rd}(\mathcal N^-)$ is the unresolved quartet tree and the taxa
form a B-quartet. In case 2,  $T_{rd}(\mathcal N^-)$ is a resolved quartet tree and the
taxa form a T-quartet.
\end{lemma}

\begin{figure}[h]
\begin{center}
\includegraphics[width=4.in]{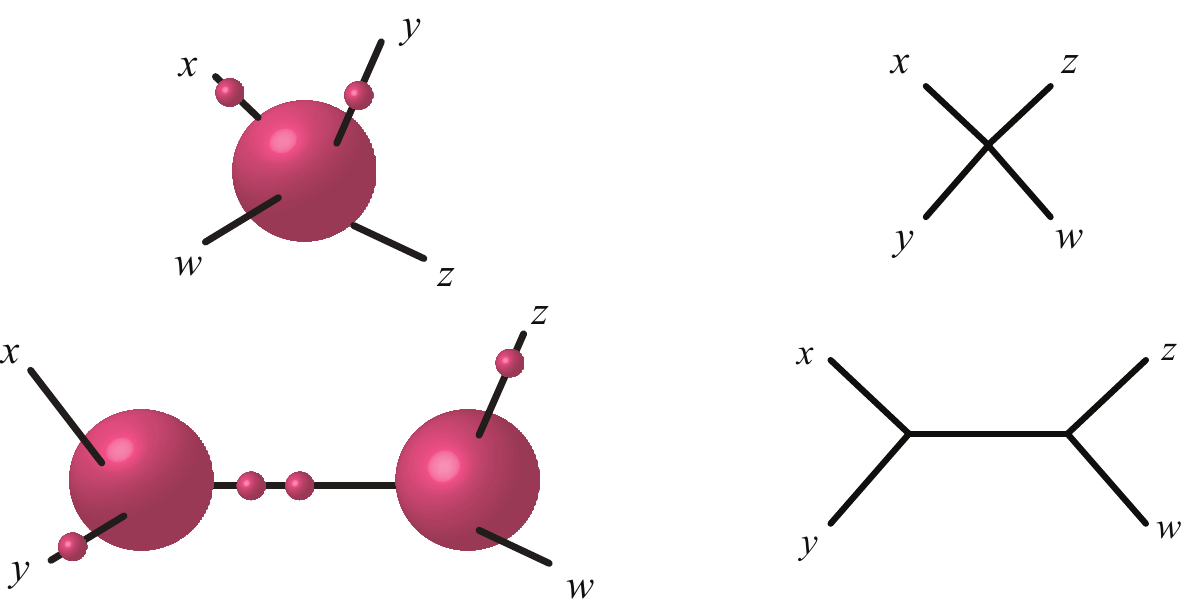}
 \caption{ (L) Schematic depictions of two semidirected unrooted  4-taxon networks $\mathcal N^-$, where spheres represent blobs of unspecified structure, and (R) their reduced unrooted trees of blobs $T_{rd}(\mathcal N^-)$. Up to taxon labelling, these are the only possible 4-taxon topological  reduced unrooted trees of blobs. }\label{fig::quartetblobs}
\end{center}
\end{figure}

As shown in \cite{Solis-Lemus2016,Banos2019}, for generic parameters on a 4-taxon level-1 network one can detect B-quartets directly from the single $CF$.
We next extend the 4-taxon result for level-1 networks to
arbitrary 4-blobs on 4-taxon networks.  

As illustrated in Figure
\ref{fig::quartetblobs}, we can determine the reduced unrooted tree of blobs of
a 4-taxon network by determining if it  has a cut edge inducing a
non-trivial split. If such a cut edge exists, the tree of blobs is a
quartet tree, and if it does not, the tree of blobs is a star
tree. That this feature can be detected by quartet concordance factors
is the content of the next proposition.

\begin{theorem} \label{thm:CFdetect} ($CF$-detectability of
  4-blobs on 4-taxon networks) Consider a 4-taxon  rooted binary
  phylogenetic network $\mathcal N^+$ on taxa $\{a,b,c,d\}$ with
  quartet concordance factor $CF=CF_{abcd}$ and reduced unrooted tree
  of blobs $T=T_{rd}(\mathcal N^-)$. Then under the NMSC for generic
  parameters:
\begin{enumerate}
\item $T$ has the quartet tree topology $ab|cd$ if, and only if, $CF_{abcd}$ is
  strictly $(ab|cd)$-cut.
\item $T$ has the unresolved quartet topology if, and only if,
  $CF_{abcd}$ is non-cut.
\end{enumerate}
\end{theorem}
\begin{proof}

We prove the following statements, for generic parameters:
\begin{itemize}
\item[(a)] If $T$ has the quartet tree topology $ab|cd$, then $CF_{abcd}$ is
  strictly $(ab|cd)$-cut.
\item[(b)] If $T$ has the unresolved quartet topology, then
  $CF_{abcd}$ is non-cut.
\end{itemize}

Were it not for the distinction between ``cut" and ``strictly cut", these statements would immediately yield claims 1 and 2.
But since the parameters are assumed to be generic, this issue is easily overcome:
Statement (b) implies for generic parameters that if $CF_{abcd}$ is cut, then $T$ has a resolved tree topology, which by (a) implies that $CF_{abcd}$ is strictly cut. Thus for generic parameters $CF_{abcd}$ is cut if, and only if, it is strictly cut. 

To establish (a),
suppose $T$ is resolved, with topology $ab|cd$. Permuting
  taxon names if necessary, we may assume that the reduced rooted tree of
  blobs thus has topology $(((a,b),c),d)$, $((a,b),(c,d))$, or $((a,b),c,d)$.

  In the first case, $(((a,b),c),d)$, if a gene tree forms under the
  NMSC by the $a,b$ lineages coalescing below the 3-blob determined by
  $a,c,d$, it contributes to the frequency of unrooted gene quartets
  with topology $ab|cd$. Otherwise, $a,b$ enter that blob as
  exchangeable lineages, and $ac|bd$ and $ad|bc$ will be equally
  probable as unrooted gene quartets. Thus $CF$ is $(ab|cd)$-cut for
  all parameters. Moreover, if the cut edges in $\mathcal N^+$ are
  given a sufficiently large length, $CF_{ab|cd}$ can be made as close
  to 1 as desired, and hence distinct from the other $CF$
  entries. Since  $CF_{abcd}$ is an analytic function of parameters and one
  parameter choice leads to its being strictly $(ab|cd)$-cut, generic
  ones must as well (since any equality of analytic functions either
  holds everywhere, or only on a lower-dimensional subset of the
  domain).

  The remaining cases, of the reduced rooted trees of blobs $((a,b),(c,d))$ and  $((a,b),c,d)$, are
  similar. Any coalescence below the blob containing the LSA leads to gene trees
  $ab|cd$. If no such coalescence occurs, then upon entering the blob containing the LSA
  the lineages from $a,b$ are
  exchangeable, so $ac|bd$ and $ad|bc$ will be equally probable gene
  tree topologies, resulting in  $CF_{abcd}$ being $(ab|cd)$-cut. Considering
  sufficiently long cut edges in $\mathcal N^+$ again shows the $CF$ is strictly
  cut generically.
 
\medskip

To prove (b), suppose $T$ has the unresolved topology, so that $\mathcal N^-$ has a 4-blob.
Again using the analyticity of  $CF_{abcd}$ it is enough to show
there is a single choice of numerical parameters that gives a non-cut $CF$.
We can even choose these parameters to be on the boundary of the
stochastic parameter space, since the analytic
parametrization of the $CF$s extends to a larger open set. We now 
show such a parameter choice exists, with some edge lengths and hybridization parameters 0.

If there are any 2-blobs on $\mathcal N^-$, set
all edge lengths in $\mathcal N^+$ that give rise to them equal to 0, with hybrid parameters
arbitrary. Doing so, we have effectively removed these blobs, and may
thus assume there are no 2-blobs  in $\mathcal N^-$. By Lemma \ref{lem:4netblobs}, the only non-leaf blob in $\mathcal N^-$ is a 4-blob.

To further simplify the network, choose some total order for the nodes in  $\mathcal N^+$ consistent with
the partial order arising from the edge directions, with the root
highest. Focus on the lowest hybrid node in this order, and its hybrid
edges $h_1,h_2$. Consider deleting one of the $h_i$  from $\mathcal N^+$ and with it all
edges from which the only directed path to a taxon leads through
$h_i$, suppressing any degree-2 nodes.  If the semidirected unrooted network of the resulting network still has a
4-blob, then set $\gamma_i=0$ and lengths for the removed edges to be
arbitrary, so that we effectively consider a network with one fewer
hybrid nodes. Its semidirected unrooted network may have 2-blobs as well as the 4-blob, but
after repeatedly `removing' 2-blobs and one of the lowest hybrid edges
in the 4-blob by setting certain parameters to 0, we arrive at a
network such that $\mathcal N^-$ still has a single 4-blob and no other blobs, but for
which removing either of $\mathcal N^+$'s lowest hybrid edges $h_1,h_2$, in
this way gives a semidirected unrooted network with no 4-blobs. We henceforth assume our
network $\mathcal N^+$ has this property.

If $v$ is the lowest hybrid node on $\mathcal N^+$, then the
subnetwork below $v$ must be a tree. But since 
$\mathcal N^-$ has no 3-blobs, this tree can only have one leaf,
and hence is a single edge. By permuting taxon names, we assume the
leaf below $v$ is labelled $a$.  Removing from $\mathcal
N^+$ either of the $h_i$, and edges above it as described
earlier, gives connected subnetworks $N_i$ which by suppressing degree-2 nodes give phylogenetic networks $\mathcal N_i^+$. Moreover,  the semidirected unrooted 
networks  $\mathcal N_i^-$ each have exactly two 3-blobs, and possibly 2-blobs. By further
permuting taxon names we may assume $\mathcal N_1^+$ has reduced
unrooted tree of blobs topology $ab|cd$.

For the sake of contradiction, suppose $\mathcal N_2^+$'s reduced
unrooted tree of blobs also has topology $ab|cd$. Consider the subnetwork $N_3$ on $b,c,d$ obtained from $\mathcal N^+$ by
deleting $a$ and all edges above $a$ that are not above any other
taxa. Then $N_3$ is a subnetwork of both $N_1$ and
$N_2$ which has a blob  $\mathcal B$ determined by the 3 taxa $b,c,d$. Let $e$ denote the cut edge of $N_3$ incident to
$\mathcal B$ through which undirected paths to $b$ pass.  
Now $e$ must be a cut edge in both $N_1$ and $N_2$, inducing the split $ab|cd$ in both.
Thus every edge in $\mathcal N^+$ which
is incident to $N_3$ and ancestral to only the taxon $a$ must be attached to
$N_3$ in the $b$-component of $N_3\smallsetminus
\{e\}$. But this implies that $e$ is a cut edge of $\mathcal
N^+$ inducing the split $ab|cd$, a contradiction to the existence of a 4-blob on $\mathcal N^-$.  Thus $\mathcal N_2^+$ has a reduced
unrooted tree of blobs topology that is resolved, but not $ab|cd$. We
henceforth assume this topology is $ac|bd$.

To pick values for the remaining parameters note that since
$a$ is the only taxon below the hybrid node $v$,
$$CF_{abcd}(\mathcal N^+)=\gamma_1 CF_{abcd}(\mathcal N_2^+)+\gamma_2 CF_{abcd}(\mathcal N_1^+).$$
where $\gamma_1, \gamma_2=1-\gamma_1$ are the hybridization parameters for $h_1,
h_2$.  Moreover, by (a) we
have that $CF(\mathcal N_1^+)$ is strictly $(ab|cd)$-cut and
$CF(\mathcal N_2^+)$ is strictly $(ac|bd)$-cut for generic parameters.
Thus by first choosing the numerical parameters other than
$\gamma_1,\gamma_2$ on $\mathcal N^+$ to yield such generic parameters
on the $\mathcal N_i^+$, we may then pick values of $\gamma_1,\gamma_2$
so that $CF_{abcd}(\mathcal N^+)$ is non-cut.  Thus
$CF_{abcd}(\mathcal N^+)$ is generically non-cut.  \qed \end{proof}

Applying this proposition to quartet $CF$s from large networks gives the following.

\begin{corollary}\label{cor:largeCFdetect} Let $\mathcal N^+$ be a metric rooted binary phylogenetic network on taxa $X$, $|X|\ge 4$, with generic
numerical parameters.
Then under the NMSC, for each 4-taxon subset $Q\subseteq X$,  the topology of the reduced unrooted tree of blobs on
  the induced network $T_{rd}(\mathcal N^-_Q)$ is identifiable from $CF_Q$.
\end{corollary}
\begin{proof} By Theorem \ref{thm:CFdetect}, for generic numerical parameter values on each induced 4-taxon
  network we have $CF$-detectability of a B-quartet or
  T-quartet.  Since
  the generic conditions only exclude a set of measure zero from the
 numerical parameter space of each 4-taxon network, they give rise to a generic
  condition on numerical parameter values on the $n$-taxon network ensuring that
  $CF$-detectability holds on all induced 4-taxon networks.   \qed\end{proof}

We now characterize more fully  the set of $CF$s that arise on
4-networks $\mathcal N^+$ whose trees of blobs are resolved.
Suppose $\mathcal N^+$ has taxa
$a,b,c,d$, and reduced unrooted tree of blobs $T_{rd}({\mathcal N^-})$ with
quartet topology $ab|cd$.  If $\mathcal N^+$ is a resolved tree, then
\cite{Allman2011} showed $CF_{ab|cd}$ may take on any value in the interval
$(1/3,1)$. If $\mathcal N^+$ is level-1, then \cite{Banos2019} showed
$CF_{ab|cd}$ may take on any value in $(1/6, 1)$. The
following generalizes these results to arbitrary networks.

\begin{proposition}\label{prop:Nk} Let $\mathcal N^+$ be a 4-taxon rooted binary phylogenetic network
  whose reduced tree of blobs has quartet topology $ab|cd$. Then under
  the NMSC the $CF$ is $ab|cd$-cut,
  with $$CF_{abcd}=(CF_{ab|cd},CF_{ac|bd},CF_{ad|bc})=(p,q,q),$$
  where $0<p,q<1$, $p+2q=1$.  Conversely, every such triple $(p,q,q)$
  arises as the $CF$ from such a network.
\end{proposition}
\begin{proof} By statement 1 of Theorem \ref{thm:CFdetect},
it only remains to establish the final claim, that every
  triple $(p,q,q)$ with $p,q>0$, $p+2q=1$, arises as the $CF$ of a
  network of the sort described.  We do this by constructing a
  sequence of topological networks $\mathcal N(k)^+$, $k\in \mathbb Z^+$, such
  that a triple $(p,q,q)$ arises as a $CF$ on $\mathcal N(k)^+$ for
  sufficiently large $k$ and certain numerical parameters.

\begin{figure}
	\begin{center}
		\includegraphics[scale=.6]{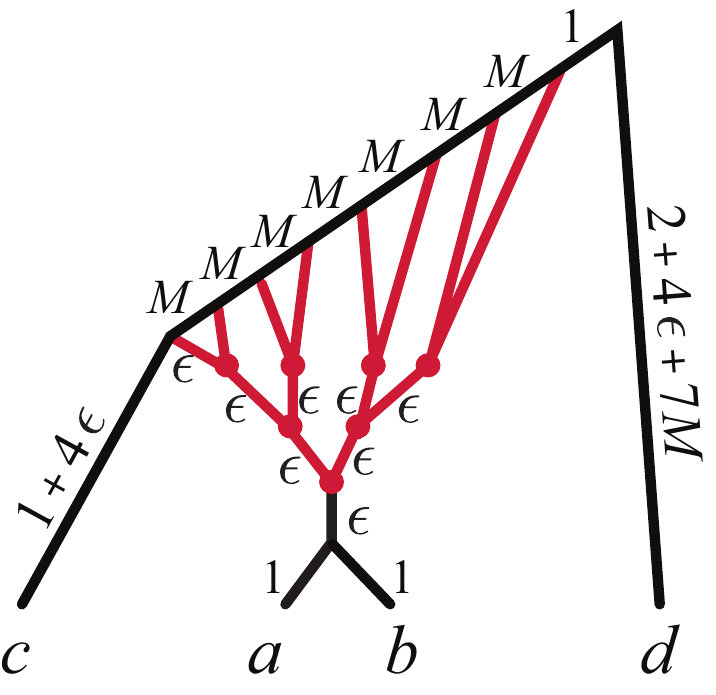}
	\end{center}
	\caption{An instance of the network $\mathcal N(k)^+$ used in the proof of Proposition \ref{prop:Nk}, 
		with $k=3$. All hybridization parameters are $1/2$, while $\epsilon$ and $M$ denote variable edge lengths. }\label{fig:Nk}
\end{figure}

The form of $\mathcal N(k)^+$ is shown in Figure \ref{fig:Nk} for
$k=3$.  Edges lead from the root to the taxon $d$ and to a 3-blob. The
other two edges incident to the 3-blob lead to $c$, and to a cherry of
$a,b$. The edge leading to the cherry has length $\epsilon$ and joins the
blob at a node that can be thought of as the `root' of an inverted binary
subtree of hybrid edges (shown in red in Figure \ref{fig:Nk}), inverted so that its edges are directed toward
this node. This binary subtree has $2^k$ `leaves', and all internal
edges of length $\epsilon$. The `pendant' edges of this subtree
have lengths $\epsilon, \epsilon+M,\epsilon+2M,\dots
,\epsilon+(2^k-1)M$, with the subtree `leaves' connected by a path of
edges all of length $M$. The pendant edges of the network $\mathcal N(k)^+$, and the
internal edge leading from the root  of $\mathcal N(k)^+$ to the 3-blob can be given any
fixed lengths, but for concreteness, we make the network ultrametric
by choosing the remaining internal edge to have length 1 and the
pendant edges to $a,b,c,d$ to be of lengths $1,1,1+(k+1)\epsilon,
2+(k+1)\epsilon +(2^k-1) M$, respectively. We set all
hybridization parameters equal to $1/2$.

Note that by Theorem \ref{thm:CFdetect}, $CF_{abcd}$ is strictly $(ab|cd)$-cut for $\epsilon>0$.

\smallskip

We next show that under the NMSC on $\mathcal N(k)^+$, for $\epsilon\approx 0$ and $k\gg 0$, with high probability the $a,b$ lineages
will be on different edges of the network
when they reach a height of $1+(k+1)\epsilon$ above the taxa.  
This event is the union of $k$ disjoint events, in which
the lineages follow the same path without coalescing to height $1+\ell\epsilon$ for any $\ell\in \{1,2,\dots, k\}$ at which point they diverge on different paths.
The probability of this for a specific $\ell$ is 
$(p/2)^{\ell},$ where $p=\exp(-\epsilon)$ is the probability two lineages do not coalesce
on an edge of length $\epsilon$.
Thus the 
probability of the full event is
\begin{equation*} 
\alpha =\sum_{\ell=1}^k (p/2)^\ell=\frac p2 \cdot \frac {1-(p/2)^k}{1-p/2}.
\end{equation*}
Taking  $\epsilon$ close to $0$ ensures $p$ is as close to 1 as desired.
Then choosing $k$ sufficiently large, the probability $\alpha$ can be made as close to $p/(2-p)$ as desired, and hence  arbitrarily close to 1.

Now $CF_{abcd}$ can be
 expressed as
$$CF_{abcd}=\alpha CF_1+(1-\alpha) CF_2,$$
where $CF_1$ is the $CF$ conditioned on the $a,b$ lineages being on
different edges at height $1+(k+1)\epsilon$ above the leaves, and $CF_2$ the $CF$
conditioned on the complementary event. To compute $CF_1$, note that
the conditioning ensures that all coalescent
events that can occur will have the same probability that they would if
they instead occurred
on a species tree with
topology $(((c,a),b),d)$ or a species tree with topology $(((c,b),a),d)$, with each of these trees
having equal probability. Moreover on these trees the length of the edge
ancestral only to the cherry is $m M$ for some $m\in \{1,2,\dots,2^k-1\}$. Thus by
choosing $M$ large enough, we can ensure with probability as close to
1 as we like that gene tree topologies will match the population tree, making
$CF_1$ as close to $(0,1/2,1/2)$ as desired. Now since
$\alpha$ can be made arbitrarily close to 1, we need not analyze $CF_2$
(beyond knowing its entries are bounded) to conclude that we can make $CF_{abcd}$
as close to $(0,1/2,1/2)$ as desired by choices of $\epsilon\approx 0$ and  $k,M\gg 0$.

Using the same fixed $k$, so the network topology is still that of
$\mathcal N(k)^+$, we could instead take $\epsilon \gg 0$, making the
probability of coalescence of $a,b$ on the edge above the $\{a,b\}$ cherry
as close to 1 as we like, so that $CF_{abcd}$ is arbitrarily close to
$(1,0,0)$. Since  $CF_{abcd}$ lies on the line of points of the form
$(q,p,p)$, $q+2p=1$ and is a continuous function of numerical
parameters, by connectedness of the numerical parameter space for
$\mathcal N(k)^+$, all intermediate points between the ones we found
arise as $CF_{abcd}$ for some parameters. \qed\end{proof} 

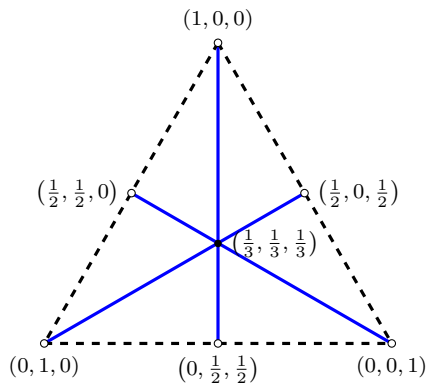
\begin{figure}
\centering
\begin{tikzpicture}
\draw[-,blue, very thick] (0,-4) -- (3.46,-2);
\draw[-,blue,very thick] (2.31,0) -- (2.31,-4);
\draw[-,blue,very thick] (4.62,-4) -- (1.16,-2);
\draw[-,dashed,very thick] (2.31,0) -- (0,-4);
\draw[-,dashed,very thick] (2.31,0) -- (4.62,-4);
\draw[-,dashed,very thick] (0,-4) -- (4.62,-4);
\draw (2.31,0) node[circle,draw=black, fill=white,inner sep=1pt,label=above:{$\left(1,0,0\right)$}]{};
\draw (2.31,-2.67) node[circle,fill,inner sep=1pt,label=right:{$\left( \frac 13,\frac 13,\frac13\right)$}]{};
\draw (0,-4) node[circle,draw=black, fill=white,inner sep=1pt,label=below:{$\left(0,1,0 \right)$}]{};
\draw (4.62,-4) node[circle,draw=black, fill=white,inner sep=1pt,label=below:{$\left(0,0,1\right)$}]{};
\draw (2.31,-4) node[circle,draw=black, fill=white,inner sep=1pt,label=below:{$\left(0,\frac{1}{2},\frac{1}{2}\right)$}]{};
\draw (1.16,-2) node[circle,draw=black, fill=white,inner sep=1pt,label=left:{$\left(\frac{1}{2},\frac{1}{2},0\right)$}]{};
\draw (3.46,-2) node[circle,draw=black, fill=white,inner sep=1pt,label=right:{$\left(\frac{1}{2},0,\frac{1}{2}\right)$}]{};
\end{tikzpicture}                
\caption{Geometric view of $CF$s for 4-taxon network models, with dashed lines outlining the simplex $\Delta^2$. The solid line segments represent $CF$s arising from species networks whose unrooted reduced trees of blobs are resolved. The vertical line segment corresponds to $ab|cd$, the upward-sloping one to $ac|bd$, and the downward sloping one to $ad|bc$.
$CF$s off of these lines can only arise from networks with unresolved unrooted reduced trees of blobs, and as shown in \cite{Banos2019} all such points arise from level-1 networks. Networks whose unrooted reduced trees of blobs are unresolved may also produce $CF$s on the line segments, but only for non-generic parameters.
}
\label{fig:simplex}
\end{figure}

The statements of Theorem \ref{thm:CFdetect} and Proposition \ref{prop:Nk}
can be made geometric by plotting $CF$s
\cite{MAR2019,Banos2019,ABR2019,AMR2022}.
A $CF$ is a point in the interior of the 2-dimensional
probability simplex,
$$\Delta^2=\left \{(p_1,p_2,p_3)\mid p_i\ge 0,\sum p_i=1\right \}.$$  
Figure \ref{fig:simplex} gives a depiction of $\Delta^2$, with the
three blue line segments within it showing the locations of cut $CF$s. 
If the unrooted reduced tree of
blobs of a 4-taxon network is $ab|cd$, then $CF_{abcd}$ lies on the
vertical line segment shown in the figure, and every
point on this line segment within the simplex arises from some such
network. The other line segments in the simplex similarly show values of
$CF_{abcd}$ arising from networks with unrooted reduced trees of blobs
$ac|bd$ and $ad|bc$.
Points in the simplex off these line segments arise as $CF$s only
for networks whose unrooted reduced trees of blobs are unresolved. By
\cite{Banos2019}, all points off the line segments arise from
level-1 networks with 4-cycles.
Although a network with a more complicated 4-blob
may produce a $CF$ on the line segments for certain numerical parameters, this cannot
happen for generic parameters by Theorem \ref{thm:CFdetect}.

\section {Blob quartet identifiability on large networks} \label{sec:BlobquartetsP}

Theorem \ref{thm:CFdetect} will be applied to the induced network
on four taxa arising from a larger $n$-taxon network. The $CF$s computed from the
induced 4-taxon networks are the same as gene tree probabilities from the large network marginalized to 4-taxon sets, by the
structure of the NMSC model. However, since four taxa which form a
B-quartet on a large network may not do so on an induced one,
determining B-quartets on a large network generally requires
additional arguments, which are developed in this section.

The following lemma leads to one easy deduction of
B-quartets from those on induced networks.

\begin{lemma} \label{lem:inducedB} Consider a network $N$ with degree-1 nodes bijectively labelled by $X$, and 
a subnetwork $ M$ of $N$ with the restricted labelling of some degree-1 nodes by $Y\subseteq X$.  If a set $S\subseteq Y$ determines a blob on $M$, then $S$ determines a blob on $ N$. Moreover, the incident cut edges of the blob on $N$ leading to elements of $S$ are in $M$.
\end{lemma}

\begin{proof}
If  $S$ determines a blob $\mathcal B_0$ on $M$, then there exist undirected paths in $M$ from $\mathcal B_0$ to each $s\in S$, with no edges in common among any pair of paths. But $\mathcal B_0$ is contained in a blob $\mathcal B$ of $N$. For each $s\in S$, the path from $\mathcal B_0$ to $s$ may include some edges in $\mathcal B$, but it has a subpath from $\mathcal B$ to $s$ entirely outside of $\mathcal B$. Moreover, these subpaths for different $s$ have no edges in common, and must thus pass through distinct cut edges incident to $\mathcal B$. Hence $S$ determines $\mathcal B$, and the incident cut edges leading to each $s$ are in $M$.
\qed\end{proof}

To apply this to induced phylogenetic networks on subsets of taxa, observe that induced networks are obtained from subnetworks
by suppressing degree-2 nodes. Under this operation, blobs pass to blobs, and cut edges to cut edges. Thus we have the following.

\begin{corollary} \label{cor:inducedB} Let $\mathcal N^+$ be a rooted binary phylogenetic network on $X$, and  $\mathcal M^+$ the induced network on $Y\subset X$. Then any
B-quartet on $\mathcal M^+$ is a B-quartet on $\mathcal N^+$.
\end{corollary}

To identify additional B-quartets from those identified by
Theorem \ref{thm:CFdetect} and Corollary \ref{cor:inducedB}, we develop
an inference rule. To state it concisely, we say
taxa $a,b$ are \emph{separated} in a
 resolved quartet if they lie in different cherries.   
Thus the
  taxa $a,b$ are separated in $ac|bd$ and $ad|bc$, but are not
  separated in $ab|cd$.

\begin{theorem} \label{thm:infrule} (B-quartet Inference Rule)
  Consider a rooted binary phylogenetic network $\mathcal N^+$ on $n$ taxa,
  $n\ge 5$.  Suppose that $\{a,b,c,d\}$ and $\{b,c,d,e\}$ are
  B-quartets on $\mathcal N^+$. If on the induced 4-taxon network any one of
  $\{a,b,c,e\}$, $\{a,b,d,e\}$, or $\{a,c,d,e\}$ is
\begin{quote}
\begin{itemize}
\item[(a)] \label{case1} a T-quartet, with $a,e$ separated in
the reduced unrooted tree of blobs for the induced 4-taxon network, or
\item[(b)] \label{case2} a B-quartet,
\end{itemize}
\end{quote}
then all of $\{a,b,c,e\}$, $\{a,b,d,e\}$, and $\{a,c,d,e\}$ are
B-quartets on $\mathcal N^+$.
\end{theorem}
\begin{proof} The taxa $b,c,d$ determine a blob in $\mathcal N^+$,
  corresponding to a node $v$ in its tree of blobs. But since $\{a,b,c,d\}$
  and $\{b,c,d,e\}$ are B-quartets, undirected paths in the tree of blobs
  from the taxa $a$ and $e$ also first meet those from $b,c,d$ at
  $v$. The conclusion will follow from showing the paths from $a$ and
  $e$ to $v$ do not meet each other before $v$, so that all 5 paths
  from $a,b,c,d,e$ first meet at $v$.

  Suppose the paths from $a,e$ do meet before $v$. Then there is an
  edge in the tree of blobs, and hence a cut edge in the network,
  that separates $a,e$ from $b,c,d$. This implies that picking
  any two of $b,c,d$, the taxa $a,e$ are not separated in the
4-taxon tree of blobs, nor do they form a B-quartet with $a,e$.  \qed\end{proof}

For example, for the network of Figure \ref{fig::BTquartets} both $\{a,b,c,d\}$ and $\{b,c,d,e\}$ are $CF$-detectable B-quartets.
While $\{a,b,d,e\}$ is not a $CF$-detectable B-quartet,  since $CF_{abde}$ is strictly $ab|de$-cut, applying Theorem \ref{thm:infrule} shows that it is a B-quartet.

\smallskip

To show that the previous propositions give sufficient tools to detect
all B-quartets for generic parameters, we use the
following lemma.

\begin{figure}
\begin{center}
\includegraphics[width=8cm]{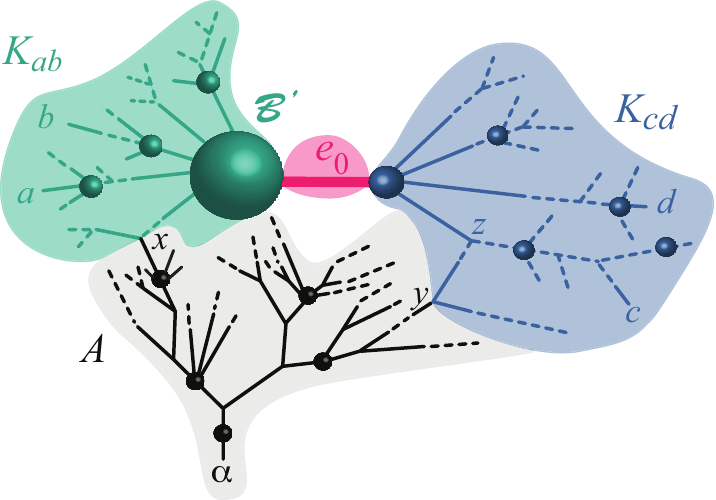}
\end{center}
\caption{A schematic of the network $\mathcal N^+$, as described in Lemma $\ref{lem:main}$. Edges are 
partitioned into four color-coded sets.   
Black edges are ancestral to the taxon $\alpha$ 
and no other taxa, forming the subnetwork $A$. Non-black edges form the
subnetwork $N'$,  in which the blob $\mathcal B'$ is determined by $\{a,b,c\}$.
The red edge $e_0$ incident to $\mathcal B'$ is a cut edge of $\mathcal N'$, separating the connected components $K_{ab}$ and $K_{cd}$, shown in green and blue, respectively. The root of $\mathcal N^+$ might be in either $K_{ab}$ or $K_{cd}$.
The nodes $x,y,z$ are described in the proof of the lemma. }\label{fig:lemmain}
 \end{figure}

\begin{lemma} \label{lem:main}
Let $\mathcal N^+$ be a rooted binary phylogenetic network on taxa $X$ for which $\{a,b,c,d\}$ is a 
B-quartet, and suppose for some $\alpha\in X$  the induced network $\mathcal N'$ on $X\smallsetminus \{\alpha\}$ has a cut edge in $\mathcal N'$ separating $a,b$ from $c,d$.
Then $\{a,b,c,\alpha\}$ is a B-quartet on the induced network $\mathcal M^+$ on $X\smallsetminus \{d\}$. 
\end{lemma}

\begin{proof} 
We define several subnetworks of $\mathcal N^+$, with Figure \ref{fig:lemmain} provided to assist the reader.
Let $A$ be the connected subnetwork of $\mathcal N^+$ whose edges are those ancestral to the taxon $\alpha$ but no other taxa. Let  $N'$ be the connected subnetwork of $\mathcal N^+$ whose edges are those ancestral to at least one taxon other than $\alpha$,
and $M$ the connected subnetwork of $\mathcal N^+$ whose edges are ancestral to at least one taxon other than $d$.
Note that $N', M$ yield the induced networks $\mathcal N',\mathcal M^+$ on $X\smallsetminus \{\alpha\}, X\smallsetminus \{d\}$ by suppressing degree-2 nodes, and $\mathcal N^+=N'\cup A$. 
 
Let  $\mathcal B'$ be the blob in $ N'$ determined by $a,b$, and $c$, and let $e_0$ be the cut edge of $N'$ incident to $\mathcal B'$ through which paths to $c$ pass. Thus $e_0$ also separates $a,b,$ from $c,d$ in $N'$. Let $K_{ab}$ (respectively $K_{cd}$) denote the connected component of $N'\smallsetminus\{e_0\}$ containing $a,b$ (respectively $c,d$). Then the edges of the four connected subnetworks $A$, $K_{ab}$, $\{e_0\},$  and $K_{bc}$ partition the edges of $\mathcal N^+$, as shown in Figure \ref{fig:lemmain}. We now construct a cycle in $\mathcal N^+$ through these subnetworks with certain features.

First, there is an undirected path $P_1$ entirely within $A$ from  $\alpha$ to a node $x$ in $K_{ab}$. If this were not the case, then all paths from  $\alpha$ to $N'$ within $A$ would end at nodes in $K_{cd}$. But then $e_0$ would separate $a,b$ from $c,d,\alpha$ in $\mathcal N^+$, contradicting that $\{a,b,c,d\}$ is a B-quartet on $\mathcal N^+$.

There is also a path $P_2$ from $x$ to $e_0$ in $K_{ab}$, by the connectedness of $K_{ab}$.

By a similar argument to that for $P_1$, there is an undirected path from a node $y\in K_{cd}$ to $\alpha$ within $A$. Because $y \in K_{cd}$, $y$ is ancestral to a taxon other than $\alpha$. Since the edge in the path incident to $y$ is ancestral only to $\alpha$, that edge's parent node must be $y$ and there is a directed path from $y$ to $\alpha$ within $A$. Choose some such directed path. 

The nodes $y$ and $c$ must have a common ancestor in $K_{cd}$, since either the root of $\mathcal N'$ is in $K_{cd}$ or any directed path from the root to any node in  $K_{cd}$ passes through $e_0$ and the child node of $e_0$ is such an ancestor. Choosing $z$ 
as a least common ancestor of $y,c$ in $K_{cd}$ (i.e., a common ancestor with no descendent that is a common ancestor), and a directed path from $z$ to $y$, we form a combined directed path $P_4$ from $z$ through $y$ to $\alpha$, with all edges ancestral to $\alpha$. 

If all edges and nodes ancestral only to $d$ are deleted from $N'$, the network remains connected and contains both $z$  and $e_0$. Thus there is a path  $P_3$ from $e_0$ to $z$ in $K_{cd}$ with no edges that are ancestral to only the taxon $d$.

Combining the paths $P_1$, $P_2$,  the edge $e_0$, $P_3$, and  $P_4$, and removing edges to eliminate any self-intersections, yields a cycle $C$ in $\mathcal N^+$ which passes through $A$, $K_{ab}$, $e_0$, and $K_{cd}$. 
This cycle also lies in $M$, as none of its edges are ancestral only to the taxon $d$. Since $\mathcal B'$ also lies in $M$, and the cycle $C$ and $\mathcal B'$ intersect, they lie in the same blob $\mathcal B$ of $M$.

It remains to show that $\{a,b,c, \alpha\}$ determines $\mathcal B$ in $M$, and hence is a B-quartet on $\mathcal M^+$.
Since $\{a,b,c\}$ determines $\mathcal B'$ in $N'$  and hence in $M\cap N'$, by Lemma \ref{lem:inducedB} $\{a,b,c\}$ determines $\mathcal B$ in $M$ with incident cut edges leading to $a,b,c$ in $M\cap N'\subset N'$.  But the initial segment of $P_1$ gives a path from $\alpha$ to $C$. Since this path lies entirely in $A$, the cut edge incident to $\mathcal B$ that leads to $\alpha$ must be in $A$, and is therefore distinct from those to $a,b,c$. Thus $\{a,b,c,\alpha\}$
determines $\mathcal B$ in $M$.\qed \end{proof}

We arrive at the main result of this section.

\begin{theorem} \label{thm:infer} On an $n$-taxon rooted binary phylogenetic
  network $\mathcal N^+$ with generic numerical parameters, all B-quartets can
  be identified from the quartet $CF$s using $CF$-detectability
  (Theorem \ref{thm:CFdetect}) and applications of the B-quartet
  Inference Rule (Theorem \ref{thm:infrule}).
\end{theorem}

\begin{proof} By Corollary \ref{cor:largeCFdetect}, for generic
  parameters we may identify the topologies of the reduced unrooted
  trees of blobs of all induced networks on four taxa. Since the
  B-quartet Inference Rule does not depend on parameters, we have  only
  to show that this information together with the inference rule is enough to identify all
  B-quartets.

  We proceed by induction on the number $n$ of taxa on the network
  $\mathcal N^+$, with the base case of $n=4$ established. Inductively
  assume that the result holds for networks with fewer than $n$ taxa,
  and consider $\mathcal N^+$ with $n\ge 5$ taxa.

  Suppose $\{a,b,c,d\}$ is a B-quartet on $\mathcal N^+$, determining a blob $\mathcal B$. Then consider the connected components of the graph obtained by deleting $\mathcal B$. Choose one taxon
 from each component which contains a taxon, with four of these being
  $a,b,c,d$. Passing to the induced network on those taxa, all edges in $\mathcal B$ are retained. If this network has fewer than $n$ taxa, then the
  inductive hypothesis gives that $\{a,b,c,d\}$ can be identified as
  a B-quartet on it, and by Corollary \ref{cor:inducedB} on
  $\mathcal N^+$.

  If the number of taxa was not decreased, then $\mathcal N^+$ has a
  relatively simple structure: Its LSA network $\mathcal N^\oplus$ contains the blob $\mathcal B$ with $n$ incident cut edges.
  If the LSA is in $\mathcal B$, then 
the incident cut edges connect to (possibly empty) chains of 2-blobs leading to leaves.  
 If the LSA is not  in $\mathcal B$, then
$n-1$ incident cut edges connect to  chains of 2-blobs leading to leaves, and one connects through a chain of 2 blobs to a 3-blob containing the LSA, which connects to another chain of 2-blobs leading to a leaf.
  For a network $\mathcal
  N^+$ of this form if we remove any taxon other than $a,b,c,d$ and pass to the
  induced network, $\{a,b,c,d\}$ either remains a B-quartet or  does
  not. If $\{a,b,c,d\}$ remains a B-quartet, then we may delete that
  taxon, and again obtain the result from the inductive hypothesis.

  Suppose then that no taxon 
  can be removed from $\mathcal N^+$ without $\{a,b,c,d\}$ ceasing to be a B-quartet in the induced network, and
  fix some $\alpha \in X \smallsetminus \{a, b, c, d\}$.
Let $\mathcal N'$ be the induced rooted network on
  $X\smallsetminus \{\alpha\}$.
  Then the blob $\mathcal B$ on $\mathcal
  N^+$ splits into multiple blobs with cut edges joining them on
  $\mathcal N'$, and $\{a,b,c,d\}$ is a T-quartet on $\mathcal N'$.
There must be a cut edge $e_0$ in $\mathcal N'$ that separates two of $a,b,c,d$, say $a,b$,
  from the others, $c,d$.
  Theorem \ref{thm:CFdetect} thus shows
  that $CF_{abcd}$ is $(ab|cd)$-cut. 

  Applying Lemma \ref{lem:main} twice, we conclude that
  $\{a,b,c,\alpha\}$ and $\{b,c,d,\alpha\}$ are B-quartets on the
  networks induced from $\mathcal N^+$ by removing $d$ and $a$
  respectively.  As these are networks on $n-1$ taxa, the inductive
  hypothesis ensures that they can be detected as
  B-quartets. But they must then also be B-quartets on $\mathcal N^+$
  by Corollary \ref{cor:inducedB}. An application of Inference
  Rule (a) of Theorem \ref{thm:infrule} then establishes the
  claim. \qed
\end{proof}

Although the proof of Theorem \ref{thm:infer} shows that only part (a) of
Theorem \ref{thm:infrule} is needed to infer all B-quartets from
those that are $CF$-detectable, part (b) is useful in an inference
algorithm for reducing computational time.

\smallskip

\begin{figure}
	\begin{center}
\includegraphics[width=4.in]{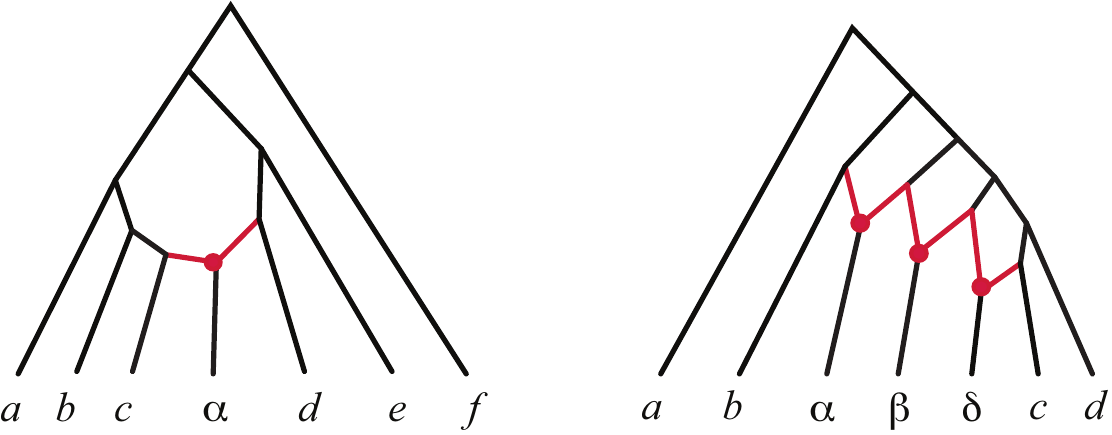}
\caption{(L) A 7-blob with a simple cycle structure. While many of its
  B-quartets are not $CF$-detectable, each can be inferred from
  $CF$-detectable ones by a single application of the B-quartet
  Inference Rule. For instance, $\{a,b,c,d\}$ is a B-quartet although
  $CF_{abcd}$ is $ad|bc$-cut. The inference rule shows that it is a
  B-quartet using the two $CF$-detectable ones, $\{\alpha, a,c,d \}$
  and $\{\alpha,b,c,d\}$. (R) A 7-blob with a more complex
  structure. The B-quartet $\{a,b,c,d\}$ is not $CF$-detectable, but
  three applications of the inference rule allow it to be inferred
  from those that are.}\label{fig:inferProof}
\end{center}
\end{figure}

\medskip

Figure \ref{fig:inferProof} shows several instructive examples of blobs for 
understanding the proof and application of Theorem
\ref{thm:infer}. On the left, a simple 7-cycle relates $a,b,c,d,e,f$
and hybrid taxon $\alpha$. Though $\{a,b,c,d\}$ is a B-quartet, it is
not $CF$-detectable since $CF_{abcd}$ is ${ad|bc}$-cut. To infer that
$\{a,b,c,d\}$ is a B-quartet by the argument of the proof, the taxa
$e,f$ can be ignored, as passing to the induced network without
them leaves $\{a,b,c,d\}$ a B-quartet. The taxon $\alpha$ will be
used, since its deletion would make $\{a,b,c,d\}$ a T-quartet.  The
$CF$s for $\{\alpha, a,c,d \}$ and $\{\alpha,b,c,d\}$ show those sets
are B-quartets, so using that $CF_{abcd}$ is ${ad|bc}$-cut in
inference rule (a) of Theorem \ref{thm:infrule} gives the desired
conclusion. Every other B-quartet for this network can be similarly
inferred using the inference rule once.

A more complicated example with a 7-blob in Figure
\ref{fig:inferProof}(R), illustrates the need for the inductive argument
for Theorem \ref{thm:infer}. Here we explain how to infer that $\{a,b,c,d\}$ is a
B-quartet even though $CF_{abcd}$ is $(ab|cd)$-cut. Note that deletion
of any one of $\alpha, \beta,\delta$ would give an induced network
with $\{a,b,c,d\}$ a T-quartet. We pick any one of these, say $\alpha$, and find that
$\{a,b,c,\alpha\}$ is a B-quartet using its $CF$.  Then, by
considering the induced 6-taxon network on $\{b,c,d,\alpha,\beta,\delta\}$,
which has a 6-blob when unrooted, we see inductively that
$\{\alpha,b,c,d\}$ is a B-quartet, so by the inference rule $\{a,b,c,d\}$ is also.
Tracing through the full argument
for $\{b,c,d,\alpha,\beta,\delta\}$ to explicitly show $\{\alpha,b,c,d\}$ is a B-quartet requires several more applications of the inference rule. 

Of course in an inference algorithm, where the network structure is
not yet known, this analysis is done in the opposite order, by
first finding all $CF$-detectable B-quartets, and then using repeated
applications of the rule to infer new ones until no more can be
produced.

\section{Main result}\label{sec:main}

The identifiability of the tree of blobs of a species network now follows easily.

\begin{theorem} \label{thm:ToBident}
Let $\mathcal N^+$ be a rooted binary phylogenetic network. Then for generic numerical parameters, the reduced unrooted tree of blobs  $T_{rd}(\mathcal N^-)$ is identifiable from the distribution of gene quartet topologies under the NMSC model.
\end{theorem}

\begin{proof} 
By Theorem \ref{thm:infer}, for generic numerical parameters on $\mathcal N^+$, all B-quartets on $\mathcal N^+$ can be identified from the quartet $CF$s, that is, from the distributions of gene quartet topologies. By Corollary \ref{cor:largeCFdetect}, we can additionally identify the topology of each unrooted reduced quartet tree of blobs if it is resolved.

Treating B-quartets on $\mathcal N^+$ as unresolved, we thus can identify the topology of every displayed quartet tree on $T_{rd}(\mathcal N^-)$. But the collection of displayed quartets determine the tree \cite{Semple2005,Rhodes2020}, so the tree of blobs is identifiable. \qed \end{proof}

This result addresses the theoretical question of whether it is in principle possible to infer $T_{rd}(\mathcal N^-)$ from quartet $CF$s, but its proof also suggests an algorithm for inference of the tree of blobs from data. This will be more completely developed in a subsequent publication, but we outline the steps here.

For a set $\{a,b,c,d\}$ of four taxa, a \emph{quartet count
  concordance factor (qcCF)} is a vector of counts
$(n_{ab|cd},n_{ac|bd}, n_{ad|bc} )$ of unrooted topological quartet
trees. We assume for each given set $\{a,b,c,d\}$ these counts summarize a sample of independent draws under the NMSC.  
For instance, these could be displayed quartets on a collection of independent gene trees on the full set $X$ of taxa, or on subsets of $X$.
While gene trees are not empirically observable, given gene sequence data they may be inferred by standard phylogenetic methods, at the price of introducing inference error.

Beginning with a collection of independent gene trees on $X$, the algorithm proceeds as follows:
\begin{enumerate}
\item Tabulate qcCFs for all sets of four taxa.
\item \label{step:test} Apply a statistical hypothesis test to each qcCF to judge whether the T-quartet model can be rejected. If so, the taxa form a putative B-quartet on the induced 4-taxon network. If not, infer the resolved quartet tree of blobs topology.
\item \label{step:B} Use the B-quartet Inference Rule repeatedly to determine all putative B-quartets on the full network.
\item \label{step:tob} Treating putative B-quartets as unresolved quartet trees and T-quartets as resolved, estimate the unrooted reduced tree of blobs, using the quartet intertaxon distance \cite{Rhodes2020} and a tree-building method.
\end{enumerate}

This algorithm is similar in outline to NANUQ \cite{ABR2019}, which provides for statistically-consistent inference of a network provided it is level-1. However, since it does not attempt to infer any details of blob structure, it avoids the complications of interpreting splits graphs.

Note that step \ref{step:test} requires the development of a novel statistical test similar to the T3 test of \cite{MAR2019}, since the cut model has a singularity at the point $(1/3,1/3,1/3)$. Step \ref{step:B} cannot be done naively, since its computational complexity needs to be controlled for use on large networks. Finally, while numerous methods exist for determining a tree from its displayed quartets, an attractive one here is to use the quartet intertaxon distance of \cite{Rhodes2020} combined with a tree building method such as Neighbor-Joining as a means of addressing potential noise in the quartets and still achieving reasonable runtimes.

\medskip

The tree of blobs shown to be identifiable by Theorem \ref{thm:ToBident}, and estimated by the algorithm sketched above, is of course a topological tree. Researchers might prefer a metric tree of blobs, indicating (in coalescent units) the distance between blobs. For edges between trivial blobs, it is straightforward to see that edge lengths are identifiable, and heuristics such as those used by ASTRAL \cite{SayyariMirarab2016} for species tree inference provide a fast estimate of them. However, if either, or both, endpoints of an edge are in non-trivial blobs, then both identifiability and methods for effective estimation are far from obvious. For example, on the 4-taxon network in Figure \ref{fig:Nk} if the length $\epsilon$ of the edge between the blobs  is held fixed  but $k$ or $M$ varied, the $CF$ varies over a line segment. This segment intersects a similar segment for a nearby value of $\epsilon$. Thus the distance between blobs cannot be identified in this 4-taxon case. This identifiability question for larger networks will also be studied in a future work.

\section{Acknowledgements}

This work was supported by the National Science Foundation, grant  2051760, awarded to JR and EA. EA and JM were also supported by
 NIGMS Institutional Development Award (IDeA), grant 2P20GM103395. HB was supported by the Moore-Simons Project on the Origin of the Eukaryotic Cell, Simons Foundation grant 735923LPI (DOI: https://doi.org/10.46714/735923LPI) awarded to Andrew J. Roger and Edward Susko.


\bibliographystyle{plain}
\bibliography{Hybridization}

\end{document}